\documentclass[12pt,letter]{amsart} 
\usepackage{amsmath,amssymb,amsthm,amscd,amsxtra,amsfonts,mathrsfs,graphicx,enumerate,bm,slashed,array} 
\usepackage[all,cmtip]{xy}
\input xy 
\xyoption{all}
\newtheorem{Theorem}{Theorem}[section]
\newtheorem{Lemma}[Theorem]{Lemma} 
\newtheorem{Proposition}[Theorem]{Proposition}

\newtheorem{Example}[Theorem]{Example}
\newtheorem{Remark}[Theorem]{Remark}

\newtheorem{Definition}[Theorem]{Definition}

\newtheorem{Definition/Lemma}[Theorem]{Definition/Lemma}
\newtheorem{Theorem and Definition}[Theorem]{Theorem and Definition}

\newtheorem{Assumptions}[Theorem]{Assumptions}

\newtheorem{Lemma/Definition}[Theorem]{Lemma/Definition}
\newtheorem{Lemma and Definition}[Theorem]{Lemma and Definition}
\newtheorem{Remark and notation}[Theorem]{Remark and notation}
\newtheorem{Physics Remark}[Theorem]{Physics Remark}

\newtheorem*{Theorem A}{Theorem A}
\newtheorem*{Proposition B}{Proposition B}
\newtheorem*{Theorem C}{Theorem C}
\newtheorem*{Theorem D}{Theorem D}
\newtheorem*{Theorem E}{Theorem E}
\newtheorem{Definition and Lemma}[Theorem]{Definition and Lemma}
\advance\evensidemargin-.5in
\advance\oddsidemargin-.5in
\advance\textwidth1in

\begin{document}
 \author{Charlie Beil}
 \address{Heilbronn Institute for Mathematical Research, School of Mathematics, Howard House, The University of Bristol, Bristol, BS8 1SN, United Kingdom.}
 \email{charlie.beil@bristol.ac.uk}
 \title{The Bell states in noncommutative algebraic geometry}
 \date{}

\begin{abstract}
We introduce new mathematical aspects of the Bell states using matrix factorizations, nonnoetherian singularities, and noncommutative blowups.
A matrix factorization of a polynomial $p$ consists of two matrices $\phi_1,\phi_2$ such that $\phi_1\phi_2 = \phi_2\phi_1 = p \operatorname{id}$. 
Using this notion, we show how the Bell states emerge from the separable product of two mixtures, by defining pure states over complex matrices rather than just the complex numbers.

We then show in an idealized algebraic setting that pure states are supported on nonnoetherian singularities.
Moreover, we find that the collapse of a Bell state is intimately related to the representation theory of the noncommutative blowup along its singular support.
This presents an exchange in geometry: the \textit{nonlocal commutative} spacetime of the entangled state emerges from an underlying \textit{local noncommutative} spacetime. 
\end{abstract}

\maketitle

\tableofcontents

\section{Introduction}

Quantum entanglement is one of the most beautiful and mysterious aspects of quantum theory, with well established experimental confirmation (notably \cite{AGR,GMR,KWWAE}).
In this paper we study the simplest form of entanglement: the Bell states.
We introduce the notion that the spacetime nonlocality inherent in an entangled pair of particles (or more generally, qubits) emerges from an underlying local geometry which is noncommutative.

We briefly outline our results.
In Section \ref{Hilbert spaces over matrix rings} we introduce a modification of quantum mechanics where the coefficient ring $\mathbb{C}$ of a complex Hilbert space 
$$\mathcal{H} \cong M_1(\mathbb{C}) \otimes_{\mathbb{C}} \mathcal{H}$$
is enlarged to the ring of matrices over $\mathbb{C}$, $$M_n(\mathbb{C}) \otimes_{\mathbb{C}} \mathcal{H}.$$
We establish density matrices, inner products, normalization, and Born's rule in this setting.

In Section \ref{pure entangled} we use this modification to factorize the Bell states using matrices. 
For example, the state $\Psi = \frac{1}{\sqrt{2}} \left( \uparrow_a \downarrow_b - \downarrow_a \uparrow_b \right) \in \mathcal{H}_a \otimes_{\mathbb{C}} \mathcal{H}_b \cong \mathbb{C}^2 \otimes_{\mathbb{C}} \mathbb{C}^2$
of two entangled particles $a$ and $b$ admits the matrix factorization
$$\left( \uparrow_a \downarrow_b - \downarrow_a \uparrow_b \right) \bm{1}_2 = \left( \begin{matrix} \uparrow_a & \downarrow_a \\ \uparrow_b & \downarrow_b \end{matrix} \right) \left( \begin{matrix} \downarrow_b & -\downarrow_a \\ -\uparrow_b & \uparrow_a \end{matrix} \right).$$
Our first main result is the following.

\begin{Theorem A} 
(Theorem \ref{Theorem 1}.)
The emergent state $\psi := \Psi \bm{1}_2$ is a separable product of two mixed states, each consisting of two pure states. 
\end{Theorem A}

In Section \ref{A new characterization of nonlocality: nonnoetherian singularities} we introduce a new algebraic characterization of entanglement in an idealized setting where spacetime is an algebraic variety $\operatorname{Max}S$ with coordinate ring $S$. 
A commutative ring is said to be noetherian if each of its ideals is finitely generated, and otherwise is nonnoetherian. 
As introduced in \cite[Section 2]{B}, a geometric space (variety or scheme) whose algebra of functions is nonnoetherian is often nonlocal, in the sense that it contains curves, surfaces, or other positive dimensional subvarieties that are single `spread-out' points.\footnote{Throughout, the term `local' will be used in the physics sense, rather than in the algebraic sense of unique maximal ideal.}
Using this property, the Einstein-Podolsky-Rosen nonlocality of $\Psi$ \cite{EPR} is captured by the nonnoetherian singularity
$$R = \mathbb{C} + I,$$
where $I \subset S$ is the ideal consisting of all polynomial functions on spacetime that vanish along the support of $\Psi$. 
The ring $R$ `sees' the support $\mathcal{Z}(I)$ as a single point since $I$ is a maximal ideal of $R$ (though $I$ is a non-maximal ideal of $S$). 

In Section \ref{Collapse from the representation theory of a noncommutative resolution} we present an exchange in geometry:
$$\textit{nonlocal commutative} \ \leadsto \ \textit{local noncommutative}$$
The exchange comes about from the \textit{noncommutative blowup} \cite[Section R]{L} of $R$ at the point $\mathcal{Z}(I)$,\footnote{A commutative blowup is an algebro-geometric process that is similar to blowing up a balloon, whereby a point ($\sim$ deflated balloon) is replaced with a projective space $\mathbb{P}^n$ ($\sim$ inflated balloon, or sphere).  Noncommutative blowups are defined differently in the context of twisted homogeneous coordinate rings.}
$$A = \operatorname{End}_R(R \oplus I) \cong \left( \begin{matrix} R & S \\ I & S \end{matrix} \right) \subseteq M_2(S).$$
This endomorphism ring may be viewed as a coordinate ring of matrix-valued functions on spacetime $\operatorname{Max}S$.
Furthermore, it replaces the nonlocal point $\mathcal{Z}(I)$ of $R$ with the set of distinct spacetime points in $\mathcal{Z}(I)$.
Indeed, denote by $d_A \in \mathbb{Z}_{\geq 0}$ the maximal dimension of the simple (i.e., irreducible) representations of $A$.
Consider the representation space 
$$\mathcal{R}(A) := \left\{ \left[ \rho \right] : A \to M_{d_A}(\mathbb{C}) \ | \ \operatorname{dim} \rho\left( \varepsilon_{11}A \varepsilon_{11} \right) = 1 \ \text{ and } \ \operatorname{dim} \rho\left( A \varepsilon_{11} \right) = d_A \right\},$$
where $\varepsilon_{11}$ is the $2 \times 2$ matrix with a $1$ in the $(1,1)$ slot and zeros elsewhere, and $[\rho]$ is the representation isoclass of $\rho$.  
In general, if an algebra $A$ is suitably nice and $\varepsilon_{11}$ is a suitable idempotent of $A$, then the representation space $\mathcal{R}(A)$ is parameterized by a commutative resolution of the center $Z(A)$ of $A$. 
We show the following. 

\begin{Proposition B} 
(Proposition \ref{sim modules}.) 
The representation space $\mathcal{R}(A)$ 
is parameterized by $\operatorname{Max}S$, and the simple representations in $\mathcal{R}(A)$ are parameterized by the open set $\mathcal{Z}(I)^c := \operatorname{Max}S \setminus \mathcal{Z}(I)$. 
\end{Proposition B}


We then introduce the following diagram to relate the matrix factorization of the Bell state $\Psi$ to the noncommutative blowup $A$ of $R$:
$$\xymatrix{ 
A \otimes_R A \ar[rr]^{\epsilon_{\mathfrak{n}} \cdot c_g} \ar[d]_{\mu} & & M_2(\mathbb{C}) \otimes_{\mathbb{C}} M_2(\mathbb{C}) \ar[r]^{\tilde{\phi}} \ar[d]_{\mu} & M_2(B) \otimes_B M_2(B) \ar[d]_{\mu} \\
A \ar[rr]^{\epsilon_{\mathfrak{n}} \cdot c_g} & & M_2(\mathbb{C}) \ar@{.>}[r] & M_2(B)
}$$
Here $\epsilon_{\mathfrak{n}} \in \mathcal{R}(A)$ is the evaluation map at a point $\mathfrak{n}$ in spacetime, $c_g$ specifies the summand ordering 
$$R \oplus I \ \ \text{ or } \ \ I \oplus R$$ 
in a matrix representation of $A = \operatorname{End}_R(R \oplus I)$, $B$ is the polynomial ring generated by $\uparrow_a$, $\downarrow_a$, $\uparrow_b$, $\downarrow_b$ over $\mathbb{C}$, and $\tilde{\phi}$ is a morphism that encodes the matrix factorization.
We will find that there does not exist a well-defined morphism $M_2(\mathbb{C}) \to M_2(B)$ that would make the diagram commute, and this lack of commutativity corresponds to the lack of uniqueness of eigenstate that the Bell state $\Psi$ may collapse onto.

Our second main result is the following, which shows how the representation theory of the noncommutative blowup $A$ characterizes the collapse of the Bell states.
In particular, the quantum randomness in the outcome of a measurement of $\Psi$ arises from the fact that there is no preference of summand ordering, $R \oplus I$ or $I \oplus R$, in $A$.

\begin{Theorem C} 
(Theorem \ref{second main theorem}.)
The emergent collapsed Bell states $\uparrow_a \downarrow_b \bm{1}_2$ and $\downarrow_a \uparrow_b \bm{1}_2$ are obtained as 1-dimensional subspaces of the full Hilbert space $\mathcal{H}_a \otimes_{\mathbb{C}} \mathcal{H}_b$,
$$\phi \epsilon_{\mathfrak{n}} c_g(A \otimes A)\mid_{\bm{1}_2} \subset \mathcal{H}_a \otimes_{\mathbb{C}} \mathcal{H}_b,$$ 
and generically only appear on the support $\mathcal{Z}(I)$ of $\Psi$:
$$\phi \epsilon_{\mathfrak{n}} c_g(A \otimes A) \mid_{\bm{1}_2} = 
\left\{ 
\begin{array}{cl} 
\uparrow_a \downarrow_b \mathbb{C} & \text{ if } \ \mathfrak{n} \in \mathcal{Z}(I) \ \text{ and } \ g = g_{R \oplus I} \\
\downarrow_a \uparrow_b \mathbb{C} & \text{ if } \ \mathfrak{n} \in \mathcal{Z}(I) \ \text{ and } \ g = g_{I \oplus R} \\
\uparrow_a \downarrow_b \mathbb{C} \ + \downarrow_a \uparrow_b \mathbb{C} & \text{ if } \ \mathfrak{n} \not \in \mathcal{Z}(I).
\end{array} \right.$$
Furthermore, the constant identity function $1 \otimes 1 \in A \otimes A$ takes the values
$$\phi \epsilon_{\mathfrak{n}} c_g(1 \otimes 1) = \left\{ \begin{array}{cl}
\uparrow_a \downarrow_b & \text{ if } \ g = g_{R \oplus I} \\
-\downarrow_a \uparrow_b & \text{ if } \ g = g_{I \oplus R}.
\end{array} \right.$$
\end{Theorem C}

\subsubsection{Notation}

We will consider two entangled qubits $a$ and $b$, such as two electrons with entangled spin (spin up and spin down), or two photons with entangled polarization (horizontal and vertical).\footnote{See \cite{BCPSW} for a general overview of entanglement.} 
Denote by $\mathcal{H}_a \cong \mathcal{H}_b \cong \mathbb{C}^2$ the respective Hilbert spaces of $a$ and $b$.
Set $\uparrow := \left( \begin{smallmatrix} 1 \\ 0 \end{smallmatrix} \right)$ and $\downarrow := \left( \begin{smallmatrix} 0 \\ 1 \end{smallmatrix} \right)$; then $\mathcal{H}_a$ and $\mathcal{H}_b$ have respective bases $\left\{ \uparrow_a, \downarrow_a \right\}$ and $\left\{ \uparrow_b, \downarrow_b \right\}$.
The Bell states are
\begin{equation} \label{1}
\Psi_{\theta} := \frac{1}{\sqrt{2}} \left( \uparrow \otimes \downarrow + e^{i \theta}\downarrow \otimes \uparrow \right) \simeq  \frac{1}{\sqrt{2}} \left( \uparrow_a \downarrow_b + e^{i \theta} \downarrow_a \uparrow_b \right) \in \mathcal{H}_a \otimes \mathcal{H}_b,
\end{equation}
and
\begin{equation} \label{2}
\Phi_{\theta} = \frac{1}{\sqrt{2}}\left( \uparrow \otimes \uparrow + e^{i \theta} \downarrow \otimes \downarrow \right) \simeq \frac{1}{\sqrt{2}}\left( \uparrow_a \uparrow_b + e^{i \theta} \downarrow_a \downarrow_b \right) \in \mathcal{H}_a \otimes \mathcal{H}_b.
\end{equation}
We will denote a general Bell state by $\Psi$.
The subscripts $a$ and $b$ allow us to symmetrize the tensor products.
In particular, it will be useful to view the Bell states as elements of the symmetric tensor algebra of $\mathcal{H}_a \oplus \mathcal{H}_b$ over $\mathbb{C}$,
\begin{equation} \label{Bequation}
B := \operatorname{Sym}\left( \mathcal{H}_a \oplus \mathcal{H}_b \right) \cong \mathbb{C}\left[ \uparrow_a, \downarrow_a, \uparrow_b, \downarrow_b \right].
\end{equation}

The Bell states possess maximal entanglement when $\theta$ equals $0$ or $\pi$. 
Furthermore, each Bell state has density matrix
\begin{equation} \label{good}
\rho = \frac 12 \left( \begin{matrix} 1 & 1 \\ 1 & 1 \end{matrix} \right).
\end{equation}
In particular, $\rho^2 = \rho$.
Thus each Bell state is pure.

We will use the term `local' in the physics sense (e.g., a wavefunction is nonlocal if it contains space-like separated points in its support), rather than in the algebraic sense (a ring is local if it contains a unique maximal ideal).
Furthermore, by nonlocality we mean \textit{quantum} nonlocality, and thus it is assumed that information cannot be transmitted faster than the speed of light.

Finally, denote by $\varepsilon_{ij} \in M_n(\mathbb{C})$ the matrix with a 1 in the $ij$th slot and zeros elsewhere.

\subsubsection{Matrix factorizations}

Eisenbud introduced the following definition in commutative algebra to study a class of singularities \cite{E}.

\begin{Definition} \rm{
A \textit{matrix factorization} of an irreducible polynomial $p \in \mathbb{C}[ x_1, \ldots, x_m]$ consists of two matrices $\phi_1, \phi_2 \in M_n\left( \mathbb{C}[ x_1, \ldots, x_m] \right)$ such that $\phi_1 \phi_2 = \phi_2 \phi_1 = p \cdot \textbf{1}_n$.\footnote{More precisely, setting $R := \mathbb{C}[ x_1, \ldots, x_m]$, a matrix factorization is a pair of $R$-module homomorphisms 
$P_1 \substack{\substack{\phi_1 \\ \longrightarrow} \\ \substack{\longleftarrow \\ \phi_2}} P_2$, 
where $P_1$ and $P_2$ are free right $R$-modules \cite[Definition 2.1]{KST}.}

Two matrix factorizations $(\phi_1,\phi_2)$ and $(\phi'_1, \phi'_2)$ of $p$ are \textit{isomorphic} if there are invertible matrices $s_1,s_2 \in \operatorname{GL}_n\left( \mathbb{C} \right)$ such that $\phi'_1 = s_2^{-1} \phi_1 s_1$ and $\phi'_2 = s_1^{-1} \phi_2 s_2$.
} \end{Definition}

\begin{Example} \rm{
Consider the matrix factorization of the polynomial $xy - zw \in \mathbb{C}\left[ x,y,z,w \right]$,
$$(xy - zw)\textbf{1}_2 = \left( \begin{matrix} x & z\\ w & y \end{matrix} \right) \left( \begin{matrix} y & -z \\ -w & x \end{matrix} \right) = \left( \begin{matrix} y & -z \\ -w & x \end{matrix} \right) \left( \begin{matrix} x & z\\ w & y \end{matrix} \right).$$
More generally, set $\xi := e^{i(\theta + \pi)/2}$.
Then
\begin{equation} \label{mf2}
(xy + e^{i \theta} zw)\textbf{1}_2 = \left( \begin{matrix} x & \xi  z\\ \xi w & y \end{matrix} \right) \left( \begin{matrix} y & -\xi z \\ -\xi w & x \end{matrix} \right).
\end{equation}
We will use these factorizations to study the Bell states (\ref{1}) and (\ref{2}), by replacing the variables $x,y,z,w$ with the spin states $\uparrow_a, \downarrow_b, \downarrow_a, \uparrow_b$, respectively $\uparrow_a, \uparrow_b, \downarrow_a, \downarrow_b$.
}\end{Example}

\begin{Remark} \rm{
The Dirac equation is an example of a matrix factorization of the Klein-Gordon equation:
$$\left(i \slashed{\partial} - m \right)\left(i \slashed{\partial} + m \right) \psi = \left( \partial^2 - m^2 \right) \psi.$$
In particular, its polynomial form with $m = 0$ is
$$\left( t^2 - x^2 - y^2 - z^2 \right) \bm{1}_4 = \left( \gamma_0 t + \gamma_1 x + \gamma_2 y + \gamma_3 z \right)^2,$$
where $\gamma_0 = \left( \begin{matrix} \bm{1}_2 & 0 \\ 0 & -\bm{1}_2 \end{matrix} \right)$ and $\gamma_i = \left( \begin{matrix} 0 & \sigma_i \\ -\sigma_i & 0 \end{matrix} \right)$ for $i =1,2,3$.
}\end{Remark}

\section{Hilbert spaces over matrix rings} \label{Hilbert spaces over matrix rings}

In this section, we introduce a modification of quantum mechanics, where the ground field $\mathbb{C}$ is replaced by $M_n(\mathbb{C})$, the algebra of $n \times n$ matrices over $\mathbb{C}$.

\begin{Definition} \rm{
Fix a finite dimensional Hilbert space $\mathcal{H}$ with basis $\left| 1 \right\rangle, \ldots, \left| m \right\rangle$, an integer $n \geq 1$, and set $\tilde{\mathcal{H}} := M_n(\mathbb{C}) \otimes_{\mathbb{C}} \mathcal{H}$.
Consider an element
$$\psi = \sum_{i = 1}^m c_i \left| i \right\rangle \in \tilde{\mathcal{H}}$$ 
with coefficients $c_i$ in $M_n(\mathbb{C})$ and $\left| i \right\rangle$ in $\mathcal{H}$.
We define the \textit{density matrix} 
$$\rho \in M_m\left( M_n(\mathbb{C}) \right)$$ 
of $\psi$, with respect to the ordered basis $\left| 1 \right\rangle, \ldots, \left| m \right\rangle$, to have entries
$$\rho_{ij} := c_i c_j^{\dagger} \in M_n(\mathbb{C}).$$
We say $\psi$ and $\rho$ are \textit{normalized} if the full trace of $\rho$ is $1 \in \mathbb{C}$, and \textit{partially normalized} if the partial trace of $\rho$ is the identity matrix $\bm{1}_n \in M_n(\mathbb{C})$.
We call $\psi$ \textit{emergent} if it is partially normalized and proportional to $\bm{1}_n$.
}\end{Definition} 

We introduce the following inner product on $\tilde{\mathcal{H}}$.

\begin{Lemma}
The pairing $\left\langle \cdot \mid \cdot \right\rangle: \tilde{\mathcal{H}} \times \tilde{\mathcal{H}} \to \mathbb{C}$, defined by
$$\left\langle \psi \mid \phi \right\rangle := \sum_{i,j} \operatorname{tr}(c_i^{\dagger}d_j) \left\langle i \mid j \right\rangle$$
for $\psi = \left| \psi \right\rangle = \sum_i c_i \left| i \right\rangle$ and $\phi = \left| \phi \right\rangle = \sum_i d_i \left| i \right\rangle$ in $\tilde{\mathcal{H}}$, is a sesquilinear positive-definite inner product on $\tilde{\mathcal{H}}$.
\end{Lemma}

\begin{proof}
Consider $\left| \psi \right\rangle = \sum_i c_i \left| i \right\rangle$ for some $c_i \in M_n(\mathbb{C})$ and $\left| i \right\rangle \in \mathcal{H}$.
Let $\left\{ \left| \ell \right\rangle \right\}_{\ell}$ be an orthonormal basis for $\mathcal{H}$.
Then for each $i$ we may write $\left| i \right\rangle = \sum_{\ell} a_{i \ell} \left| \ell \right\rangle$ with $a_{i \ell} \in \mathbb{C}$.
Furthermore, we may write $c_i = \sum_{1 \leq s,t \leq n} \gamma_{ist} \varepsilon_{st}$ with $\gamma_{st} \in \mathbb{C}$, whence
$$\left| \psi \right\rangle = \sum_{i,s,t,\ell} a_{i \ell} \gamma_{ist} \varepsilon_{st} \left| \ell \right\rangle.$$
Therefore
$$\left\langle \psi \mid \psi \right\rangle = \sum_{i,s,t,\ell} \operatorname{tr}\left( \left( a_{i \ell} \gamma_{ist}\varepsilon_{st} \right)^{\dagger} a_{i \ell} \gamma_{ist}\varepsilon_{st} \right) \left\langle \ell \mid \ell \right\rangle
= \sum_{i,s,t,\ell} \left| a_{i \ell}\gamma_{st} \right|^2 \geq 0,$$
with equality if and only if $\left| \psi \right\rangle = 0$.
Furthermore, it is straightforward to check that $\left\langle \psi \mid \phi \right\rangle = \left\langle \phi \mid \psi \right\rangle^*$. \end{proof}

\begin{Lemma}
A state $\psi = \sum_i c_i \left| i \right\rangle \in \tilde{\mathcal{H}}$ satisfies
$$\sum_i \operatorname{tr}(c_i^{\dagger}c_i) = 1 \in \mathbb{C}$$
if and only if $\psi$ is normalized.
\end{Lemma}

\begin{proof}
Let $\rho$ be the density matrix of $\psi$.
Then
$$\operatorname{tr}\left(\rho \right) = \sum_i \operatorname{tr}(c_i c_i^{\dagger} ) = \sum_i \operatorname{tr}(c_i^{\dagger}c_i).$$
\end{proof}

\begin{Remark} \rm{
\textit{Generalized Born rule.}
If $\psi = \sum_i c_i \left| i \right\rangle$ is the normalized wavefunction for a particle written in terms of an eigenbasis $\left\{ \left| i \right\rangle \right\}_i$ for $\mathcal{H}$, then we postulate that the probability of finding a particle in the state $\left| i \right\rangle$ is $\operatorname{tr}(c^{\dagger}_i c_i )$.
This of course reduces to the usual Born probability $|c_i|^2 = c^*_i c_i$ in the case $n = 1$, where $M_1(\mathbb{C}) \otimes_{\mathbb{C}} \mathcal{H} \cong \mathcal{H}$.
}\end{Remark}

\begin{Definition} \rm{
Consider two finite dimensional Hilbert spaces $\mathcal{H}_1$ and $\mathcal{H}_2$.
An $M_n(\mathbb{C}),M_n(\mathbb{C})$-bimodule structure is defined on $\tilde{\mathcal{H}}_i$ by setting
$$\beta (\alpha \otimes h) := \beta \alpha \otimes h \ \ \text{ and } \ \ (\alpha \otimes h) \beta := \alpha \beta \otimes h,$$
for $h \in \mathcal{H}_i$, $\alpha, \beta \in M_n(\mathbb{C})$.
We call an element
$$\psi \in \tilde{\mathcal{H}}_1 \otimes_{M_n(\mathbb{C})} \tilde{\mathcal{H}}_2$$
\textit{separable} if it can be written as a product $\psi = \phi_1 \otimes \phi_2$ with $\phi_i \in \tilde{\mathcal{H}}_i$, and \textit{entangled} otherwise.
}\end{Definition}

\section{From pure entangled to mixed separable via matrix factorizations} \label{pure entangled}

In this section we analyze the Bell states as emergent states with coefficients in $M_2(\mathbb{C})$.

\begin{Theorem} \label{Theorem 1}
Let $\Psi$ be a Bell state $\Psi_{\theta}$ or $\Phi_{\theta}$ as in (\ref{1}) and (\ref{2}).
Then the emergent state $\psi := \Psi \bm{1}_2$ is a separable product of two mixed states, each consisting of two pure states. 
\end{Theorem}

\begin{proof}
Fix $\theta \in [0, 2 \pi ]$ and set $\xi  := e^{i(\theta + \pi)/2}$.
It suffices to consider the Bell state $\Psi := \Psi_{\theta} = \frac{1}{\sqrt{2}} \left( \uparrow_a \downarrow_b + e^{i \theta} \downarrow_a \uparrow_b \right)$.
By (\ref{mf2}), $\Psi$ admits the matrix factorization
$$\Psi \bm{1}_2 = \frac{1}{\sqrt{2}} \left( \uparrow_a \downarrow_b + e^{i \theta} \downarrow_a \uparrow_b \right) \textbf{1}_2 = \frac{1}{\sqrt{2}} \left( \begin{matrix} \uparrow_a & \xi  \downarrow_a \\ \xi  \uparrow_b & \downarrow_b \end{matrix} \right) \left( \begin{matrix} \downarrow_b & -\xi  \downarrow_a \\ -\xi  \uparrow_b & \uparrow_a \end{matrix} \right).$$ 
The emergent state $\Psi \bm{1}_2$ is then proportional to the product of the normalized states
\begin{equation} \label{psi}
\phi_1 := \frac 12 \left( \begin{matrix} \uparrow_a & \xi  \downarrow_a \\ \xi  \uparrow_b & \downarrow_b \end{matrix} \right) = \frac 12 \left( \varepsilon_{11} \uparrow_a + \xi  \varepsilon_{12} \downarrow_a + \xi  \varepsilon_{21} \uparrow_b + \varepsilon_{22} \downarrow_b \right)
\end{equation}
and
\begin{equation} \label{phi}
\phi_2 := \frac 12 \left( \begin{matrix} \downarrow_b & -\xi  \downarrow_a \\ -\xi  \uparrow_b & \uparrow_a \end{matrix} \right) = \frac 12 \left( \varepsilon_{11}\downarrow_b -\xi  \varepsilon_{12}\downarrow_a -\xi  \varepsilon_{21}\uparrow_b + \varepsilon_{22}\uparrow_a \right),
\end{equation}
which are elements of $\tilde{\mathcal{H}}_a \oplus \tilde{\mathcal{H}}_b \subset M_2(B)$.
Specifically,
$$\psi := \frac{1}{\sqrt{2}} \Psi \bm{1}_2 = \frac{1}{2} \left( \varepsilon_{11} + \varepsilon_{22} \right) \left( \uparrow_a \downarrow_b + e^{i \theta} \downarrow_a \uparrow_b \right) = 2 \phi_1 \phi_2 = 2 \phi_2 \phi_1 \in M_2(B).$$

\begin{Remark} \rm{
The commutation $\phi_1 \phi_2 = \phi_2 \phi_1$ generalizes the commutation 
$$\uparrow_a \downarrow_b \ = \ \downarrow_b \uparrow_a \ \ \text{ and } \ \ \downarrow_a \uparrow_b \ = \ \uparrow_b \downarrow_a$$
in the symmetrization of the tensor product $\mathcal{H}_a \otimes_{\mathbb{C}} \mathcal{H}_b$.\footnote{This symmetrization allows us to view $\uparrow_a \downarrow_b$ and $\downarrow_a \uparrow_b$ as elements in the symmetric tensor algebra $B = \mathbb{C}\left[ \uparrow_a, \downarrow_a, \uparrow_b, \downarrow_b \right]$ of $\mathcal{H}_a \oplus \mathcal{H}_b$ over $\mathbb{C}$.}
}\end{Remark}

The normalized density matrix of $\psi$ with respect to the ordered basis $\left\{ \uparrow_a \downarrow_b, \downarrow_a \uparrow_b \right\}$ is
$$\rho_{\psi} = \frac 14 \left( \varepsilon_{11} + \varepsilon_{22} \right) \left( \begin{matrix} 1 & 1 \\ 1 & 1 \end{matrix} \right),$$
and its partially normalized density matrix is
$$\hat{\rho}_{\psi} := 2 \rho_{\psi}.$$
Thus
$$\rho_{\psi}^2 = \frac 14 \rho_{\psi} \not = \rho_{\psi} \ \ \ \text{ and } \ \ \ \hat{\rho}_{\psi}^2 = \hat{\rho}_{\psi}.$$
Therefore $\psi$ is mixed when normalized, and pure when partially normalized. 
In other words, \textit{$\psi$ appears to be pure when viewed as an emergent state, but is really a mixture when its internal degrees of freedom--its matrix components--are taken into account.}

Since $\psi$ is a mixture, we would like to determine what pure states it is composed of.
For $i = 1,2$, set
$$\psi_i := \frac{1}{\sqrt{2}} \varepsilon_{ii} \left( \uparrow_a \downarrow_b + e^{i \theta} \downarrow_a \uparrow_b \right).$$
Then similar to (\ref{good}), $\psi_i$ has density matrix
$$\rho_{\psi_i} = \frac 12 \varepsilon_{ii} \left( \begin{matrix} 1 & 1 \\ 1 & 1 \end{matrix} \right).$$ 
Thus $\rho_{\psi_i}^2 = \rho_{\psi_i}$, and therefore $\psi_i$ is pure.
It follows that
$$\rho_{\psi} = \frac 12 \left| \psi_1 \right\rangle \left\langle \psi_1 \right| + \frac 12 \left| \psi_2 \right\rangle \left\langle \psi_2 \right|.$$
Therefore $\psi$ is a mixture of the two pure states $\psi_1$ and $\psi_2$, and these states occur with equal probability.

We now analyze the states $\phi_1$ and $\phi_2$.
Using (\ref{psi}) and (\ref{phi}), their normalized density matrices with respect to the ordered basis $\left\{ \uparrow_a, \downarrow_a, \uparrow_b, \downarrow_b \right\}$ are
$$\rho_{\phi_1} = 
\frac 14 \left( \begin{matrix}
\varepsilon_{11} & 0 & \xi^* \varepsilon_{12} & 0 \\
0 & \varepsilon_{11} & 0 & \xi \varepsilon_{12}\\
\xi \varepsilon_{21} & 0 & \varepsilon_{22} & 0 \\
0 & \xi^* \varepsilon_{21} & 0 & \varepsilon_{22}
\end{matrix} \right), \ \
\rho_{\phi_2} = \frac 14 \left( \begin{matrix} 
\varepsilon_{22} & -\xi^*  \varepsilon_{21} & 0 & 0\\
-\xi \varepsilon_{12} &\varepsilon_{11} & 0 & 0 \\
0 & 0 & \varepsilon_{22} & -\xi  \varepsilon_{21}\\
0 & 0 & -\xi^* \varepsilon_{12} & \varepsilon_{11}
\end{matrix} \right).$$
Their partially normalized density matrices are
$$\hat{\rho}_{\phi_i} := 2 \rho_{\phi_i}.$$
Thus
$$\rho_{\phi_i}^2 = \frac 12 \rho_{\phi_i} \not = \rho_{\phi_i} \ \ \text{ and } \ \ \hat{\rho}_{\phi_i}^2 = \hat{\rho}_{\phi_i}.$$
Therefore $\phi_1$ and $\phi_2$ are mixed when normalized, and pure when partially normalized.

Since $\phi_1$ and $\phi_1$ are mixed states, we would like to determine what pure states they are composed of, as before.
Consider the states constructed from the columns of $\phi_1$,
$$\eta_{11} := \frac{1}{\sqrt{2}} \left( \varepsilon_{11} \uparrow_a + \xi  \varepsilon_{21} \uparrow_b \right), \ \ \ \ \ \ \eta_{12} := \frac{1}{\sqrt{2}} \left( \xi \varepsilon_{12} \downarrow_a + \varepsilon_{22} \downarrow_b \right),$$
and the columns of $\phi_2$,
$$\eta_{21} := \frac{1}{\sqrt{2}} \left( -\xi \varepsilon_{21} \uparrow_b + \varepsilon_{11} \downarrow_b \right), \ \ \ \ \ \ \eta_{22} := \frac{1}{\sqrt{2}} \left(\varepsilon_{22} \uparrow_a - \xi  \varepsilon_{12} \downarrow_a \right).$$
It is straightforward to check that their normalized density matrices satisfy
$$\rho_{\eta_{ij}}^2 = \rho_{\eta_{ij}},$$
and so each $\eta_{ij}$ is a pure state. 
For example,
$$\rho_{\eta_{11}} = \frac 12 \left( \begin{matrix} \varepsilon_{11} & 0 & \xi^* \varepsilon_{12} & 0 \\ 0 & 0 & 0 & 0 \\ \xi  \varepsilon_{21} & 0 & \varepsilon_{22} & 0 \\ 0 & 0 & 0 & 0 \end{matrix} \right) 
= \frac{1}{2} \left( \begin{matrix} \varepsilon_{11} \\ 0 \\ \xi  \varepsilon_{21} \\ 0 \end{matrix} \right) \left( \begin{matrix} \varepsilon_{11} & 0 & \xi^*  \varepsilon_{12} & 0 \end{matrix} \right) = \left| \eta_{11} \right\rangle \left\langle \eta_{11} \right|.$$
Furthermore, for $i =1,2$ we have
$$\rho_{\phi_i} = \frac 12 \left| \eta_{i1} \right\rangle \left\langle \eta_{i1} \right| + \frac 12 \left| \eta_{i2} \right\rangle \left\langle \eta_{i2} \right|.$$
Therefore $\phi_i$ is an ensemble consisting of the two pure states $\eta_{i1}$ and $\eta_{i2}$, and these states occur with equal probability,
$$\phi_i = \frac{1}{\sqrt{2}} \left( \eta_{i1} + \eta_{i2} \right) = \frac{1}{\sqrt{2}} \left( \begin{array}{c|c} \eta_{i1} & \eta_{i2} \end{array} \right).$$
\end{proof}

In future work it would be interesting to consider matrix factorizations of higher level and multiparticle entanglement.

\section{An algebraic framework} \label{noncommutative resolutions section}

\subsubsection{Preliminaries}

We begin by recalling some elementary algebraic geometry.

To any commutative algebra $S$ containing $\mathbb{C}$ we may associate a geometric space $\operatorname{Max}S$.
The points of $\operatorname{Max}S$ are nonzero algebra homomorphisms $\rho: S \to \mathbb{C}$, or equivalently, their kernels $\operatorname{ker}\rho$.
These homomorphisms are the simple (i.e., irreducible) representations of $S$, and their kernels are the maximal ideals of $S$, the set of which is also denoted $\operatorname{Max}S$.\footnote{In a commutative ring, the maximal and primitive ideal spectra coincide \cite[Proposition 2.15]{GW}.  Here we are focusing on primitive ideals, or closed points, rather than prime ideals, because in the next section we will be interested in the geometry that arises from the simple representations of a noncommutative algebra.}
$S$ may be viewed as a ring of functions on $\operatorname{Max}S$: for each $f \in S$ and simple representation $\rho$ of $A$ with kernel $\mathfrak{n} = \operatorname{ker} \rho \in \operatorname{Max}S$, set 
\begin{equation} \label{f(n)}
f(\mathfrak{n}) := \rho(f) \in \mathbb{C},
\end{equation}
or equivalently, 
$$f(\mathfrak{n}) := f + \mathfrak{n} \in S/\mathfrak{n} \cong \mathbb{C}.$$ 

Associate to each set $Y$ of $\operatorname{Max}S$ the ideal
$$\mathcal{I}(Y) := \left\{ f \in S \ | \ f \in \mathfrak{n} \text{ for each } \mathfrak{n} \in Y \right\} \subset S,$$
which is the set of functions in $S$ that vanish identically on $Y$.
Conversely, associate to each ideal $J$ of $S$ the subset
$$\mathcal{Z}(J) := \left\{ \mathfrak{n} \in \operatorname{Max}S \ | \ \mathfrak{n} \supseteq J \right\} \subset \operatorname{Max}S,$$
which is the common zero locus of the functions in $J$.
These subsets form the closed sets of a topology on $\operatorname{Max}S$, called the Zariski topology. 

If $S$ is a finitely generated $\mathbb{C}$-algebra with no nonzero nilpotent elements, then by Hilbert's Nullstellensatz \cite[Proposition 1.12, Corollary 1.47]{GW2}, for any closed set $Y \subset \operatorname{Max}S$ and ideal $J \subset S$ satisfying
\begin{equation} \label{radical}
J = \left\{ f \in S \ | \ f^n \in J \text{ for some } n \geq 1 \right\} =: \sqrt{J},
\end{equation}
we have
$$\mathcal{Z}\left(\mathcal{I}(Y) \right) = Y \ \ \text{ and } \ \ \mathcal{I}\left(\mathcal{Z}(J) \right) = J.$$
In this case $S$ and $\operatorname{Max}S$ uniquely determine each other up to isomorphism,\footnote{An algebra homomorphism $h: S \to S'$ determines a morphism $\operatorname{Max}S' \to \operatorname{Max}S$ by sending the point $\rho': S' \to \mathbb{C}$ in $\operatorname{Max}S'$ to the point $\rho := \rho' h: S \stackrel{h}{\to} S' \stackrel{\rho'}{\to} \mathbb{C}$ in $\operatorname{Max}S$.} 
$\operatorname{Max}S$ is called an algebraic variety, and $S$ is called its coordinate ring.

\subsection{A new characterization of nonlocality: nonnoetherian singularities} \label{A new characterization of nonlocality: nonnoetherian singularities}

In this section we present a new characterization of quantum nonlocality as a nonnoetherian singularity birationally equivalent to spacetime.  
We begin by introducing the following definitions.

\begin{Definition} \label{support} \rm{
The \textit{real support} of a wavefunction $\Psi$ is the locus of events in spacetime where it is possible in principle to measure $\Psi$, 
while its \textit{instrumental support} is the locus of events where $\Psi$ is actually measured. 
}\end{Definition}

\begin{Remark} \label{epistemic} \rm{ 
Recall that an ontic state is a state of reality, while an epistemic state is a state of knowledge.
The definition of real support fits into the framework of both ontic and epistemic realist quantum theories, and merely specifies the events where a measuring apparatus could be placed so that $\Psi$ may be measured.
In contrast, the definition of instrumental support fits into the framework of epistemic non-realist quantum theories.
}\end{Remark}

As a mathematical toy model, we make the following assumptions. 

\begin{Assumptions} \rm{ \ 
\begin{itemize}
 \item The complexification of spacetime is a smooth algebraic variety $X = \operatorname{Max}S$ with coordinate ring $S$.
 \item Emergent pure states are supported on (Zariski) closed subsets of $X$.
\end{itemize}
} \end{Assumptions}

For example, we may take spacetime to be flat, in which case $X = \mathbb{C}^4$ and $S = \mathbb{C}[x,y,z,t]$.

\begin{Definition} \label{supporting cr} \rm{
We define the \textit{supporting coordinate ring} of a pure state $\Psi$ with support $Y \subset X$ to be the subalgebra 
$$R = \mathbb{C} + \mathcal{I}(Y) \subset S,$$ 
where $I := \mathcal{I}(Y)$ is the radical ideal of $Y = \mathcal{Z}(I)$.
}\end{Definition}

Recall that a commutative ring is noetherian if each of its ideals is finitely generated, and otherwise is nonnoetherian. 

\begin{Lemma} \label{nn remark} \rm{(\cite[Corollary 2.22]{B})}
Let $I$ be an ideal of a finitely generated $\mathbb{C}$-algebra $S$.
Then the ring $R = \mathbb{C} + I$ is nonnoetherian if and only if $\operatorname{dim}\mathcal{Z}(I) \geq 1$.
\end{Lemma}

The following lemma shows that Definition \ref{supporting cr} captures what Einstein called `spooky action at a distance', which is a fundamental property of any pure state supported on more than one point of a spatial slice of spacetime.
Recall that two varieties are birational if they are isomorphic on nonempty open subsets.

\begin{Lemma} (The spooky lemma.)
Let $\Psi$ be a pure state with support $\mathcal{Z}(I) \subset X$, and let $R = \mathbb{C} + I \subset S$ be its supporting coordinate ring.
Then $\operatorname{Max}R$ coincides with $\operatorname{Max}S$ except that the locus $\mathcal{Z}(I) \subset \operatorname{Max}S$ is identified as one single `spread-out' point in $\operatorname{Max}R$.  
In particular, the locus
$$U := \left\{ \mathfrak{n} \in \operatorname{Max}S \ | \ R_{\mathfrak{n} \cap R} = S_{\mathfrak{n}} \right\} \subset X = \operatorname{Max}S$$
equals with the complement of $\mathcal{Z}(I)$ in $\operatorname{Max}S$.

Consequently, the possibly nonnoetherian singularity $\operatorname{Max}R$ is birational to the algebraic variety $\operatorname{Max}S$.
\end{Lemma}

\begin{proof}
$I$ is clearly a maximal ideal of $R$, and so $I$ is a closed point of $\operatorname{Spec}R$.
The claim that $U = \operatorname{Max}S \setminus \mathcal{Z}(I)$ follows from \cite[Proposition 2.8]{B}, and birationality follows from \cite[Theorem 2.5.3]{B}.
\end{proof}

We note that the birational morphism $\operatorname{Max}S \to \operatorname{Max}R$, $\mathfrak{n} \mapsto \mathfrak{n} \cap R$, is in general not proper.

Now let $\Psi$ denote an entangled Bell state, and let us assume in our toy model that Alice's particle $a$ and Bob's particle $b$ are point-like. 
Consider flat complexified spacetime $X = \mathbb{C}^4$ with coordinate ring $S = \mathbb{C}[x,y,z,t]$.
Suppose that the entangled particles are traveling at a constant speed $v$ in the $z$ direction away from each other, relative to their center-of-mass frame.
The real support $Y_{\operatorname{re}}$ of $\Psi$ is then the zero locus
$$Y_{\operatorname{re}} = \left\{ x = y = z - vt = 0 \right\} \cup \left\{ x = y = z + vt = 0 \right\} \subset X.$$ 
$Y_{\operatorname{re}}$ has radical ideal (\ref{radical}),\footnote{An ideal generated by elements $g_1, \ldots, g_n$ in $S$ is denoted $(g_1,\ldots, g_n) := g_1S + \cdots + g_nS$.}
$$I_{\operatorname{re}} := \mathcal{I}(Y_{\operatorname{re}}) = \sqrt{\left( x, y, z - vt \right) \left( x,y, z + vt \right)} = xS + yS + (z-vt)(z+vt)S.$$

Further suppose Alice and Bob each measure their respective particles at the spacetime events 
$$p_a := (x_a,y_a,z_a,t_a) \ \ \text{ and } \ \ p_b := (x_b,y_b,z_b,t_b)$$
in $X$.
The instrumental support $Y_{\operatorname{in}}$ of $\Psi$ is then the union
$$Y_{\operatorname{in}} = p_a \cup p_b.$$
$Y_{\operatorname{in}}$ has radical ideal
$$I_{\operatorname{in}} := \mathcal{I}(Y_{\operatorname{in}}) = \sqrt{ \mathfrak{n}_a \mathfrak{n}_b },$$
where $\mathfrak{n}_{a} = \left( x-x_a, y-y_a, z - z_a,t-t_a \right)$ and $\mathfrak{n}_{b} = \left( x-x_b, y-y_b, z - z_b,t-t_b \right)$ are the maximal ideals of $S$ consisting of all functions that vanish at the respective events $p_a$ and $p_b$.

By Lemma \ref{nn remark}, the real supporting coordinate ring $R_{\operatorname{re}} = \mathbb{C} + I_{\operatorname{re}}$ of $\Psi$ is nonnoetherian since its real support $Y_{\operatorname{re}}$ is 1 (complex) dimensional. 
In contrast, the instrumental supporting coordinate ring $R_{\operatorname{in}} = \mathbb{C} + I_{\operatorname{in}}$ of $\Psi$ is noetherian since its instrumental support $Y_{\operatorname{in}}$ is 0 dimensional.

\subsection{Collapse from the representation theory of a noncommutative blowup} \label{Collapse from the representation theory of a noncommutative resolution}

Without loss of generality we will consider the Bell state $\Psi = \ \uparrow_a \downarrow_b - \downarrow_a \uparrow_b$. 
Fix a type of support, real or instrumental, and denote by $R = \mathbb{C} + I \subset S$ the supporting coordinate ring for $\Psi$ as in Section \ref{A new characterization of nonlocality: nonnoetherian singularities}.

Throughout, given a left $R$-module $M = M_1 \oplus \cdots \oplus M_{\ell}$ with each $M_i$ indecomposable, we will denote by $\operatorname{end}_R\left(M_1 \oplus \cdots \oplus M_{\ell} \right)$ the matrix ring whose $ij$-th entry is $\operatorname{Hom}_R\left(M_j, M_i \right)$; this matrix ring is isomorphic to the endomorphism ring $\operatorname{End}_R(M)$ by fixing a particular basis.\footnote{An endomorphism ring, like a linear transformation, is defined without reference to a particular basis.}

The noncommutative blowup of $R$ along the support $\mathcal{Z}(I)$ of $\Psi$ is the endomorphism ring
$$\operatorname{End}_R(R \oplus I) \cong \operatorname{end}_R(R \oplus I) 
:= \left( \begin{matrix} \operatorname{Hom}_R(R,R) & \operatorname{Hom}_R(I,R) \\ \operatorname{Hom}_R(R,I) & \operatorname{Hom}_R(I,I) \end{matrix} \right)
\cong \left( \begin{matrix} R & S \\ I & S \end{matrix} \right) \subseteq M_2(S).$$
The algebra $A := \operatorname{end}_R(R \oplus I)$ is a modification of $R$ in the sense that $R \cong \operatorname{End}_R(R)$, and its center is isomorphic to $R$,
$$Z(A) = R \bm{1}_2.$$
(In the introduction we took $A$ to be $\operatorname{End}_R(R \oplus I)$ rather than $\operatorname{end}_R(R \oplus I)$, for ease of exposition.)

Furthermore, $A$ may be viewed as a noncommutative coordinate ring on the spacetime variety $\operatorname{Max}S$: the evaluation of a function $f \in A$ at a point $\mathfrak{n} \in \operatorname{Max}S$ is the image of $f$ under the representation
\begin{equation} \label{A rep}
\epsilon_{\mathfrak{n}}: A \longrightarrow \left( \begin{matrix} R/\left( \mathfrak{n}\cap R\right) & S/\mathfrak{n} \\ I/ \left( \mathfrak{n} \cap I \right) & S/\mathfrak{n} \end{matrix} \right)
\cong \left\{ \begin{array}{cl}
M_2(\mathbb{C}) & \text{ if } \ \mathfrak{n} \not \in \mathcal{Z}(I) \\
\left( \begin{matrix} \mathbb{C} & \mathbb{C} \\ 0 & \mathbb{C} \end{matrix} \right) & \text{ if } \ \mathfrak{n} \in \mathcal{Z}(I), \end{array} \right.
\end{equation}
that is, 
$$f( \mathfrak{n}) := \epsilon_{\mathfrak{n}}(f) \in M_2(\mathbb{C}).$$
Note that this is analogous to the commutative case (\ref{f(n)}).
However, in the following proposition we show that a representation $\rho: A \to M_2(\mathbb{C})$ is simple if and only if $\rho$ is isomorphic to $\epsilon_{\mathfrak{n}}$ for some $\mathfrak{n} \not \in \mathcal{Z}(I)$, and this occurs precisely when $\mathfrak{n}$ is not in the support of $\Psi$.
Consequently, we will find that the Bell states only collapse on representations which are not simple.

Denote by $d_A \in \mathbb{Z}_{\geq 0}$ the maximal $\mathbb{C}$-dimension of the simple representations of $A$.
Further, given a representation $\rho: A \to M_n(\mathbb{C})$, denote by $\left[ \rho \right]$ its representation isoclass.
Consider the representation space
$$\mathcal{R}(A) := \left\{ \left[ \rho \right] : A \to M_{d_A}(\mathbb{C}) \ | \ \operatorname{dim} \rho\left( \varepsilon_{11}A \varepsilon_{11} \right) = 1 \ \text{ and } \ \operatorname{dim} \rho\left( A \varepsilon_{11} \right) = d_A \right\},$$
which appears in the study of noncommutative resolutions (under the guise of a particular stability condition).\footnote{Specifically, if an endomorphism ring of the form $A = \operatorname{End}_R(R \oplus M)$ is a noncommutative resolution of its singular center 
$$Z(A) \cong R \cong \varepsilon_{11}A\varepsilon_{11},$$
then $\mathcal{R}(A)$ is often parameterized by a commutative resolution of $R$.}

\begin{Proposition} \label{sim modules} 
The representation space $\mathcal{R}(A)$ is parameterized by the spacetime variety $\operatorname{Max}S$, and the simple representations in $\mathcal{R}(A)$ are parameterized by the open set $\mathcal{Z}(I)^c := \operatorname{Max}S \setminus \mathcal{Z}(I)$. 
\end{Proposition}

\begin{proof}
We first claim that $d_A = 2$.
Indeed, since the corner rings $\varepsilon_{11}A\varepsilon_{11} \cong R$ and $\varepsilon_{22}A\varepsilon_{22} \cong S$ are commutative algebras over the algebraically closed field $\mathbb{C}$, any simple representation $\rho$ of $A$ over $\mathbb{C}$ will be at most two dimensional and satisfy
$$\operatorname{dim} \rho\left( \varepsilon_{11}A\varepsilon_{11} \right) = \operatorname{dim} \rho\left( \varepsilon_{22} A \varepsilon_{22} \right) = 1.$$

Now suppose $\rho: A \to M_2(\mathbb{C})$ is a representation whose isoclass is in $\mathcal{R}(A)$.
The conditions $\operatorname{dim} \rho\left( \varepsilon_{11}A \varepsilon_{11} \right) = 1$ and $\operatorname{dim} \rho\left( A \varepsilon_{11} \right) = 2$ imply that $\operatorname{dim} \rho\left( \varepsilon_{22}A \varepsilon_{22} \right) = 1$ and $\rho\left(\varepsilon_{12} \right) \not = 0$. 
Therefore, since $R$ and $S$ are commutative $\mathbb{C}$-algebras, their kernels are maximal ideals $\mathfrak{m} \in \operatorname{Max}R$ and $\mathfrak{n} \in \operatorname{Max}S$:
$$\varepsilon_{11}A\varepsilon_{11} \cong R \stackrel{\rho}{\longrightarrow} R/\mathfrak{m} \cong \mathbb{C} \ \ \text{ and } \ \ \varepsilon_{22}A\varepsilon_{22} \cong S \stackrel{\rho}{\longrightarrow} S/\mathfrak{n} \cong \mathbb{C}.$$

Furthermore, since $\rho\left( \varepsilon_{12} \right) \not = 0$ and $\rho$ is an algebra homomorphism, $\mathfrak{m} = \mathfrak{n} \cap R$.
Therefore $\rho$ is isomorphic to $\epsilon_{\mathfrak{n}}$. 

Finally, $\epsilon_{\mathfrak{n}}$ is simple if and only if $I/(\mathfrak{n} \cap I) \not = 0$, if and only if $\mathfrak{n}$ does not contain $I$.
\end{proof} 

We introduce the following notion to capture a representation-theoretic perspective of wavefunction collapse.

\begin{Definition} \rm{
A \textit{summand ordering} of the ring $\operatorname{End}_R\left(M_1 \oplus \cdots \oplus M_{\ell} \right)$ is a choice of ordering of the direct summands $M_i$.
Such a choice is specified by a permutation matrix $g \in \operatorname{GL}_{\ell}(\mathbb{C})$, which we also refer to as a summand ordering, by the change-of-basis
$$g \left( \operatorname{end}_R\left(M_1 \oplus \cdots \oplus M_{\ell} \right) \right) g^{-1} = \operatorname{end}_R \left( M_{\pi^{-1}(1)} \oplus \cdots \oplus M_{\pi^{-1}(\ell)} \right).$$
} \end{Definition}

There are two summand orderings of $\operatorname{End}_R(R \oplus I)$, namely 
$$A = \operatorname{end}_R(R \oplus I) \cong \left( \begin{matrix} R & S \\ I & S \end{matrix} \right) \ \ \text{ and } \ \ 
\operatorname{end}_R(I \oplus R) \cong \left( \begin{matrix} S & I \\ S & R \end{matrix} \right),$$
given by the respective isomorphisms of $A$,
$$g_{R \oplus I} := \bm{1}_2 \ \ \text{ and } \ \ g_{I \oplus R} := \left( \begin{matrix} 0 & 1 \\ 1 & 0 \end{matrix} \right).$$
In the following we will show that the choice of summand ordering of the $R$-module $R \oplus I$ determines what eigenstate the entangled Bell state $\Psi$ collapses onto.

Consider the evaluation representation of $\operatorname{end}_R(I \oplus R)$ at a point $\mathfrak{n} \in \operatorname{Max}S$ as in (\ref{A rep}),
$$\epsilon_{\mathfrak{n}}: \operatorname{end}_R(I \oplus R) \longrightarrow \left( \begin{matrix} S/\mathfrak{n} & I/(\mathfrak{n} \cap I) \\ S/\mathfrak{n} & R/(\mathfrak{n} \cap R) \end{matrix} \right) 
\subseteq M_2(\mathbb{C}).$$
The conjugation map $c_g$ for $g \in \left\{ g_{R \oplus I}, g_{I \oplus R} \right\}$, defined by $c_g(a) := gag^{-1}$, commutes with $\epsilon_{\mathfrak{n}}$,\footnote{This follows since 
$$\epsilon_{\mathfrak{n}} c_g(a) = \epsilon_{\mathfrak{n}}\left(g a g^{-1} \right) = \epsilon_{\mathfrak{n}}(g) \epsilon_{\mathfrak{n}}(a) \epsilon_{\mathfrak{n}}(g^{-1}) = g \epsilon_{\mathfrak{n}}(a) g^{-1} = c_g \epsilon_{\mathfrak{n}}(a).$$}
and therefore $g$ may also be viewed as a particular choice of basis for the representation $\epsilon_{\mathfrak{n}}: A \to M_2(\mathbb{C})$.

Recall the symmetric tensor algebra $B$ of $\mathcal{H}_a \oplus \mathcal{H}_b$ over $\mathbb{C}$, given in (\ref{Bequation}).  
For each spacetime point $\mathfrak{n} \in \operatorname{Max}S$ and summand ordering $g \in \left\{ g_{R \oplus I}, g_{I \oplus R} \right\}$ of $A$, consider the diagram,
\begin{equation} \label{big rep}
\xymatrix{ 
A \otimes_R A \ar[rr]^{\epsilon_{\mathfrak{n}} \cdot c_g} \ar[d]_{\mu} & & M_2(\mathbb{C}) \otimes_{\mathbb{C}} M_2(\mathbb{C}) \ar[r]^{\tilde{\phi}} \ar[d]_{\mu} & M_2(B) \otimes_B M_2(B) \ar[d]_{\mu} \\
A \ar[rr]^{\epsilon_{\mathfrak{n}} \cdot c_g} & & M_2(\mathbb{C}) \ar@{.>}[r] & M_2(B)
}
\end{equation}
Note that the tensor products are over the centers of the respective algebras.
The morphisms are defined as follows:

$\bullet$ Each vertical morphism is the multiplication map, $\mu(a_1 \otimes a_2 ) := a_1 a_2$.

$\bullet$ The representation $\epsilon_{\mathfrak{n}}c_g: A \to M_2(\mathbb{C})$ is extended to an $R,R$-bimodule homomorphism on $A \otimes_R A$ by
$$\epsilon_{\mathfrak{n}}c_g(a_1 \otimes a_2) := g \epsilon_{\mathfrak{n}}(a_1) \otimes \epsilon_{\mathfrak{n}}(a_2) g^{-1}.$$
This further extends to a representation of the full tensor algebra $\epsilon_{\mathfrak{n}} c_g: T_R(A) \to T_{\mathbb{C}}\left(M_2(\mathbb{C})\right)$, although we will not use this representation here.

$\bullet$ The morphism $\tilde{\phi}$ is the $\mathbb{C},\mathbb{C}$-bimodule homomorphism defined on the basis $\left\{ \varepsilon_{jk} \right\}$ of $M_2(\mathbb{C})$ by
\begin{equation} \label{varphi}
\varepsilon_{ij} \otimes \varepsilon_{k \ell} \mapsto \varepsilon_{ii} \phi_1 \varepsilon_{jj} \otimes \varepsilon_{kk} \phi_2 \varepsilon_{\ell \ell},
\end{equation}
with $(\phi_1,\phi_2)$ the matrix factorization of $\Psi$ given in (\ref{psi}) and (\ref{phi}).
$\tilde{\phi}$ is then extended $\mathbb{C}$-linearly to $M_2(\mathbb{C}) \otimes_{\mathbb{C}} M_2(\mathbb{C})$.
In particular,
$$\tilde{\phi} \cdot \left( \begin{matrix} 1 & 1 \\ 1 & 1 \end{matrix} \right) \otimes \left( \begin{matrix} 1 & 1 \\ 1 & 1 \end{matrix} \right) = \phi_1 \otimes \phi_2.$$
We call the composition $\phi := \mu \tilde{\phi}$ the \textit{state morphism} of $\Psi$.

\begin{Remark} \rm{
We note that the left square in (\ref{big rep}) commutes, whereas there is no morphism $M_2(\mathbb{C}) \to M_2(B)$ that would make the right square commute, since any such morphism would necessarily not be well-defined (see (\ref{1 1}) in the following theorem).
We propose that this ambiguity is what gives rise to the randomness in the outcome of a measurement of $\Psi$.
}\end{Remark}

\begin{Theorem} \label{second main theorem}
Consider the morphism $\phi$ in (\ref{varphi}) corresponding to the Bell state $\Psi$. 
The emergent collapsed eigenstates $\uparrow_a \downarrow_b \bm{1}_2$ and $\downarrow_a \uparrow_b \bm{1}_2$ are obtained as 1-dimensional subspaces of the full Hilbert space $\mathcal{H}_a \otimes_{\mathbb{C}} \mathcal{H}_b$,
$$\phi \epsilon_{\mathfrak{n}} c_g(A \otimes A)\mid_{\bm{1}_2} \subset \mathcal{H}_a \otimes_{\mathbb{C}} \mathcal{H}_b,$$ 
and generically only appear on the support $\mathcal{Z}(I)$ of $\Psi$:
\begin{equation} \label{first claim}
\phi \epsilon_{\mathfrak{n}} c_g(A \otimes A) \mid_{\bm{1}_2} = 
\left\{ 
\begin{array}{cl} 
\uparrow_a \downarrow_b \mathbb{C} & \text{ if } \ \mathfrak{n} \in \mathcal{Z}(I) \ \text{ and } \ g = g_{R \oplus I} \\
\downarrow_a \uparrow_b \mathbb{C} & \text{ if } \ \mathfrak{n} \in \mathcal{Z}(I) \ \text{ and } \ g = g_{I \oplus R} \\
\uparrow_a \downarrow_b \mathbb{C} \ + \downarrow_a \uparrow_b \mathbb{C} & \text{ if } \ \mathfrak{n} \not \in \mathcal{Z}(I).
\end{array} \right.
\end{equation}
Furthermore, the constant identity function $1 \otimes 1 \in A \otimes A$ takes the values
\begin{equation} \label{1 1}
\phi \epsilon_{\mathfrak{n}} c_g(1 \otimes 1) = \left\{ \begin{array}{cl}
\uparrow_a \downarrow_b & \text{ if } \ g = g_{R \oplus I} \\
-\downarrow_a \uparrow_b & \text{ if } \ g = g_{I \oplus R}.
\end{array} \right.
\end{equation}
\end{Theorem}

\begin{proof}
First suppose $\mathfrak{n} \in \mathcal{Z}(I)$, i.e., $\mathfrak{n} \supseteq I$.
Then $I/(\mathfrak{n} \cap I) = 0$.
If $g = g_{R \oplus I} = \bm{1}_2$, then
$$\begin{array}{rcl}
\phi \epsilon_{\mathfrak{n}} c_g (A \otimes A) & = & \phi \cdot \left( \begin{matrix} R/(\mathfrak{n} \cap R) & S/\mathfrak{n} \\ I/(\mathfrak{n} \cap I) & S/\mathfrak{n} \end{matrix} \right) \otimes \left( \begin{matrix} R/(\mathfrak{n} \cap R) & S/\mathfrak{n} \\ I/(\mathfrak{n} \cap I) & S/\mathfrak{n} \end{matrix} \right)\\
& = & \phi \cdot \left( \begin{matrix} \mathbb{C} & \mathbb{C} \\ 0 & \mathbb{C} \end{matrix} \right) \otimes \left( \begin{matrix} \mathbb{C} & \mathbb{C} \\ 0 & \mathbb{C} \end{matrix} \right)\\
& = & \left( \begin{matrix} \uparrow_a \mathbb{C} & \downarrow_a \mathbb{C} \\ 0 & \downarrow_b \mathbb{C} \end{matrix} \right) \left( \begin{matrix} \downarrow_b \mathbb{C} & \downarrow_a \mathbb{C} \\ 0 & \uparrow_a \mathbb{C} \end{matrix} \right).
\end{array}$$
Similarly, if $g = g_{I \oplus R} = \left( \begin{smallmatrix} 0 & 1 \\ 1 & 0 \end{smallmatrix} \right) = g^{-1}$, then
$$\begin{array}{rcl}
\phi \epsilon_{\mathfrak{n}} c_g (A \otimes A) & = & 
\phi \cdot g \left( \begin{matrix} \mathbb{C} & \mathbb{C} \\ 0 & \mathbb{C} \end{matrix} \right) \otimes \left( \begin{matrix} \mathbb{C} & \mathbb{C} \\ 0 & \mathbb{C} \end{matrix} \right) g^{-1}\\
& = & \phi \cdot \left( \begin{matrix} 0 & \mathbb{C} \\ \mathbb{C} & \mathbb{C} \end{matrix} \right) \otimes \left( \begin{matrix} \mathbb{C} & \mathbb{C} \\ \mathbb{C} & 0 \end{matrix} \right)\\
& = & \left( \begin{matrix} 0 & \downarrow_a \mathbb{C} \\ \uparrow_b \mathbb{C} & \downarrow_b \mathbb{C} \end{matrix} \right) \left( \begin{matrix} \downarrow_b \mathbb{C} & \downarrow_a \mathbb{C} \\ \uparrow_b \mathbb{C} & 0 \end{matrix} \right).
\end{array}$$

Now suppose $\mathfrak{n}$ is not in $\mathcal{Z}(I)$.
Then $I/(\mathfrak{n} \cap I) \cong \mathbb{C}$, and the first claim (\ref{first claim}) follows.

The second claim (\ref{1 1}) is straightforward to verify.
\end{proof}

\begin{Remark} \rm{
The roles of the orderings $R \oplus I$ and $I \oplus R$ in Theorem \ref{second main theorem} can be exchanged by considering the bimodule homomorphism $\phi'$ defined as in (\ref{varphi}) with the matrix factorization
$$\phi'_1 = \left( \begin{matrix} \uparrow_b & \downarrow_b \\ \uparrow_a & \downarrow_a \end{matrix} \right) \ \ \text{ and } \ \ \phi'_2 = \left( \begin{matrix} -\downarrow_a & \downarrow_b \\ \uparrow_a & -\uparrow_b \end{matrix} \right)$$
in place of $\phi$.
This matrix factorization is isomorphic to $(\phi_1,\phi_2)$ since the following diagram commutes,
$$\xymatrix{ 
B^{\oplus 2} \ar[r]^{\phi_{1}} \ar[d]_{\bm{1}_2} & B^{\oplus 2} \ar[r]^{\phi_{2}} \ar[d]_{\left( \begin{smallmatrix} 0 & 1 \\ 1 & 0 \end{smallmatrix} \right)} & B^{\oplus 2} \ar[d]_{\bm{1}_2} \\ 
B^{\oplus 2} \ar[r]_{\phi'_1} & B^{\oplus 2} \ar[r]_{\phi'_2} & B^{\oplus 2}
}$$
}\end{Remark}

The following are notable observations that follow from Theorem \ref{second main theorem}.
\begin{itemize}
 \item \textit{The randomness inherent in the outcome of a measurement of the Bell state $\Psi$ arises from the choice of summand ordering $R \oplus I$ or $I \oplus R$}, noting that there is no mathematical preference of one summand ordering over the other.
 \item \textit{The center $Z(A) = R \bm{1}_2$ of $A$ determines the possible emergent observed states,} by (\ref{1 1}).
 \item \textit{$\Psi$ only collapses within its support $\mathcal{Z}(I)$,} by (\ref{first claim}).  
For $\mathfrak{n} \not \in \mathcal{Z}(I)$, the dimension of the subspace $\phi \epsilon_{\mathfrak{n}} c_g(A \otimes A)\mid_{\bm{1}_2}$ is greater than one, and thus no state (eigenstate or superposition) is specified.
This corresponds to the fact that the particles $a$ and $b$ cannot be observed outside the support of $\Psi$.  
Equivalently, $\Psi$ does not collapse at representations of $A$ (points in spacetime) that are simple, by Proposition \ref{sim modules}.
\end{itemize} 

In regards to the definition of real support and Remark \ref{epistemic}, we conclude with a final observation. 
This observation explains the physical sense in which $A$ is the coordinate ring for a \textit{local} noncommutative geometry.

\begin{Remark} \rm{
Suppose the particles $a$ and $b$ are entangled at the spacetime event $\mathfrak{n}_0$, and the polarization of $a$ is measured at the event $\mathfrak{n}_1$.  
According to standard quantum mechanics, at all points $\mathfrak{n}$ along the worldline of $a$ strictly between $\mathfrak{n}_0$ and $\mathfrak{n}_1$, we should find the superposition subspace
$$\phi \epsilon_{\mathfrak{n}} c_g(A \otimes A)\mid_{\bm{1}_2} \ = \ (\uparrow_a \downarrow_b - \downarrow_a \uparrow_b) \mathbb{C},$$
while at the event $\mathfrak{n}_1$ we should find the eigenspaces
$$\phi \epsilon_{\mathfrak{n}_1} c_g(A \otimes A)\mid_{\bm{1}_2} \ = \ \uparrow_a \downarrow_b \mathbb{C} \ \ \text{ and } \ \ \downarrow_a \uparrow_b \mathbb{C},$$
depending on $g$.
However, according to Theorem \ref{second main theorem}, $\Psi$ collapses along its entire support $\mathcal{Z}(I)$.  
Thus, if we take the support to be real, then \textit{$\Psi$ collapses along the entire worldline of $a$ between $\mathfrak{n}_0$ and $\mathfrak{n}_1$}. 

In contrast, the morphism $\phi$ does not depend on $\mathfrak{n} \in \operatorname{Max}S$, and thus exists independently of spacetime.
Therefore the information of the non-collapsed state $\Psi$ continues to exist as the particles fly apart.
We are thus led to a perspective that is analogous to the de Broglie-Bohm pilot wave interpretation, where the morphism $\phi$ plays the role of the pilot wave and the representation $\epsilon_{\mathfrak{n}} c_g$ of $A$ plays the role of the particles. 
In particular, if $\Psi$ interacts \textit{reversibly} (or unitarily) with its environment, such as when a photon passes through a polarizer, then the interaction occurs with the state morphism $\phi$.
On the other hand, if $\Psi$ interacts \textit{irreversibly} with its environment, such as in a measurement of $\Psi$, then the interaction occurs with the non-faithful morphism $\epsilon_{\mathfrak{n}} c_g$. 
We leave these speculations for future work.
} \end{Remark}

\textbf{Acknowledgments.} Part of this article is based on work supported by the Simons Foundation while the author was a postdoc at the Simons Center for Geometry and Physics at Stony Brook University. 
He was also supported in part by a PFGW grant, which he gratefully acknowledges.

\bibliographystyle{hep}
\def\cprime{$'$} \def\cprime{$'$}

\end{document}